\documentclass[preprint,10pt]{elsarticle}
\usepackage{amssymb,amsmath}
\usepackage{amsthm}
\usepackage{graphicx}
\usepackage{color}
\usepackage{epstopdf}
\usepackage{subfigure}
\newcommand{\hide}[1]{}

\newtheorem{theorem}{Theorem}
\newtheorem{lemma}{Lemma}
\newtheorem{corollary}{Corollary}

\newtheorem{example}{Example}
\newtheorem{proposition}{Proposition}

\usepackage{lineno}

\usepackage{soul} 

\usepackage{amsthm}
\usepackage[linesnumbered,ruled]{algorithm2e}
\SetProcNameSty{textsc}
\SetProcArgSty{textsc}
\usepackage{paralist}
\usepackage{tikz}
\usetikzlibrary{calc,shapes,positioning}
\tikzset
{
  VSt/.style =
  {
    circle,  
    inner sep=-0.5,
    minimum size = 1.5mm,
    fill = black!10,
    draw = black,
  }
}

\tikzset
{
  SinkSt/.style =
  {
    rectangle,  
    inner sep=-0.5,
    minimum size = 1.1mm,
    fill = black!80,
    draw = black,
  }
}

\tikzset
{
  PointSt/.style =
  {
    circle,  
    inner sep=-0.5,
    minimum size = 0.8mm,
    fill = black!80,
    draw = black,
  }
}

\tikzset
{
  TickStH/.style =
  {
    rectangle,  
    inner sep=-0.5,
    minimum height = 0.5pt,
    minimum width = 1mm,
    fill = black,
    draw = black,
  }
}

\tikzset
{
  TickStV/.style =
  {
    rectangle,  
    inner sep=-0.5,
    minimum height = 1mm,
    minimum width = 0.5pt,
    fill = black,
    draw = black,
  }
}

\usepackage{float}
\newfloat{floatTogether}{th!}{ext}

\journal{Theoretical Computer Science}
\begin{document}
\begin{frontmatter}
\title{Minsum $k$-Sink Problem on Path Networks}%
\author{Robert Benkoczi,$^a$ Binay Bhattacharya,$^b$
Yuya Higashikawa,$^c$
Tsunehiko~Kameda,$^{b,1}$
and Naoki Katoh$^d$}
\address{$^a$Dept. of Math \& Computer Science, Univ. of Lethbridge, Lethbridge, Canada\\
$^b$School of Computing Science, Simon Fraser University,
Burnaby, Canada\\
$^c$School of Business Administration, Univ. of Hyogo, Kobe, Japan\\
$^d$School of Science \& Technology, Kwansei Gakuin Univ., Sanda, Japan}
\fntext[label1]{Corresponding author. email: tiko@sfu.ca, Phone: $+$1-778-887-9577.}

\begin{abstract}
We consider the problem of locating a set of $k$ sinks on a path network
with general edge capacities
that minimizes the sum of the evacuation times of all evacuees.
We first present an $O(kn\log^4n)$ time algorithm when the edge capacities are non-uniform,
where $n$ is the number of vertices.
We then present an $O(kn\log^3 n)$ time algorithm when the edge capacities are uniform.
We also present an $O(n\log n)$ time algorithm for the special case where $k=1$
and the edge capacities are non-uniform.
\end{abstract}
\begin{keyword}
Facility location\sep sink location problem\sep evacuation problem\sep minsum criterion\sep dynamic flow in network
\end{keyword}
\end{frontmatter}

\section{Introduction}
Due to many recent disasters such as earthquakes, volcanic eruptions, hurricanes,
and nuclear plant accidents,
evacuation planning is getting increasing attention.
The evacuation $k$-sink problem is an attempt to model evacuation
in such emergency situations~\cite{cheng2013,hamacher2002}.
In this paper,
a {\em $k$-sink} means a set of $k$ sinks that minimizes the sum
of the evacuation time of every evacuee to a sink.

Researchers have worked mainly on two objective functions.
One is the evacuation completion time (minmax criterion),
and the other is the sum of the evacuation times of all the evacuees
(minsum criterion).
It is assumed that all evacuees from a vertex evacuate to the same sink.

Mamada {\em et al.}~\cite{mamada2006} solved the minmax 1-sink problem
for tree networks in $O(n \log^2 n)$ time under the condition that
only a vertex can be a sink.
When edge capacities are uniform,
Higashikawa et al.~\cite{higashikawa2014b} and
Bhattacharya and Kameda~\cite{bhattacharya2015b}
presented $O(n \log n)$ time algorithms
with a more relaxed condition that the sink can be on an edge.
Chen and Golin~\cite{chen2016b} solved the minmax $k$-sink problem on
 tree networks in $O(k^2n\log^5n)$ time when the edge capacities are non-uniform.
Regarding the minmax $k$-sink on path networks,
Higashikawa {\em et al.}~\cite{higashikawa2015a} present an algorithm to compute
a $k$-sink in $O(kn)$ time if the edge capacities are uniform.
In the general edge capacity case,
Arumugam {\em et al.}~\cite{arumugam2014} showed that a $k$-sink 
can be found in $O(kn \log^2 n)$ time.
Bhattacharya et al.~\cite{bhattacharya2017a} recently improved these results to
$O(\min\{n\log n, n + k^2\log^2n)$ time in the uniform edge capacity case,
and to $O(\min\{n\log^3n, n\log n +k^2\log^4n\})$ time in the general case.
Table~\ref{tbl:survey} is a list of known algorithms and their time complexities.
\begin{table}[h]\label{tbl:survey}
\centering
\begin{tabular}{|l|l|l|}
\hline
Topology	&Problem		&Time complexity \\
\hline\hline
Path 		&1-sink [U]	&$O(n)$~\cite{higashikawa2015a}\\
\cline{2-3}
		&2-sink [U]	&$O(n)$~\cite{higashikawa2015a} \\
\cline{2-3}
		&2-sink [G]	&$O(n\log n)$~\cite{bhattacharya2017a}\\
\cline{2-3}
		&$k$-sink [U]	&$O(kn)$~\cite{higashikawa2015a}, $O(n + k^2\log^2n)$~\cite{bhattacharya2017a}, $O(n\log n)$ \cite{bhattacharya2017a}\\
\cline{2-3}
		&$k$-sink [G]	&$O(n\log n + k^2\log^4n)$~\cite{bhattacharya2017a}, $O(n\log^3 n)$ \cite{bhattacharya2017a}\\
\hline
Tree	&1-sink [U]	&$O(n\log n)$~\cite{bhattacharya2015b,higashikawa2014b}\\
\cline{2-3}
		&1-sink [G]	&$O(n\log^2n)$~\cite{mamada2004}\\
\cline{2-3}
		&$k$-sink [U]	& $O(kn^2\log^4n)$ \cite{chen2016b}, $O(\max\{k,\log n\} kn\log^3n)$ \cite{chen2014}\\
\cline{2-3}
		&$k$-sink [G]	& $O(kn^2\log^5n)$ \cite{chen2016b}, $O(\max\{k,\log n\} kn\log^4n)$ \cite{chen2014}\\
\hline
Cycle	&1-sink  [U]	&$O(n)$ \cite{benkoczi2018b}\\
\cline{2-3}
		&1-sink [G]	& $O(n\log n)$ \cite{benkoczi2018b}\\
\hline
\end{tabular}
\vspace{2mm}
\caption{Most efficient algorithms for finding completion-time sinks.
 [U] means Uniform edge capacities and [G] means General (non-uniform) edge capacities.}
\end{table}

The minsum objective function for the sink problems is motivated,
among others,
by the desire to minimize the transportation cost of evacuation
or the total amount of psychological duress suffered by the evacuees.
It is more difficult than the minmax variety because the objective cost function
is not unimodal along a path,
and, to the best of our knowledge, practically nothing is known about this problem
on more general networks than path networks.
A path network, although simple, can model an airplane aisle, a hall way in a building,
a street, a highway, etc., to name a few.
For the simplest case of $k=1$ and uniform edge capacities,
Higashikawa et al.~\cite{higashikawa2015a} proposed an $O(n)$ time algorithm.
In Sec.~\ref{sec:1sink} of this paper we present an $O(n\log n)$ time algorithm for the case of $k=1$
and general edge capacities.
For the case of general $k$ and uniform edge capacities,
Higashikawa et al.~\cite{higashikawa2015a} showed that
a $k$-sink can be found in time bounded by $O(kn^2)$ and $2^{O(\sqrt{\log k \log\log n})}n^2$.
Bhattacharya et al.~\cite{bhattacharya2018a} recently showed that 
a minsum 1-sink in path networks with uniform edge capacities that achieves
 {\em minmax regret}~\cite{kouvelis1997} can be computed in $O(n^2\log^2n)$ time.

A somewhat related problem,
the {\em quickest transshipment problem},
is defined by dynamic flows with a given set of sources and sinks:
each source has a fixed amount of supply, and each sink has a fixed demand.
The problem is to send exactly the right amount of supply from each source to each sink
with the minimum overall time.
This problem has been studied for over fifty years since the seminal work by Ford and Fulkerson~\cite{ford1958}. 
The standard technique to solve this problem is to consider discrete time steps and make a copy of the original
network for every time unit from time zero to a time horizon $T$.
This process produces a {\em time-expanded network}~\cite{ford1958}.
Gale \cite{gale1959} posed {\em the earliest arrival $s$-$t$-flow problem}, 
which is the problem of maximizing the amount of flow that reaches the single sink
by any time $\theta$, $0\le \theta \le T$.
See Wilkinson \cite{wilkinson1971}, and Minieka \cite{minieka1973}, 
Baumann and Skutella~\cite{baumann2006} for more recent results.
 
The main contributions of this paper are $O(kn\log^4 n)$ and $O(kn\log^3 n)$
time algorithms for computing a minsum $k$-sink, 
in the non-uniform and uniform edge capacity cases,
respectively.
These are improvements over our $O(kn^2\log^4n)$ and $O(kn^2\log^3n)$ time
algorithms we presented at {\sc Iwoca} 2018~\cite{benkoczi2018a}.
In achieving these results,
we introduce two innovative methods.
One is used in Sec.~\ref{sec:switchingPoint} to efficiently optimize functions formulated
as Dynamic Programming (DP),
and the other is used in Sec.~\ref{sec:Rij} to compute the cost of moving the evacuees on a subpath to
a potential sink.
We also present an $O(n\log n)$ time algorithm for the special case where $k=1$
and the edge capacities are non-uniform.

This paper is organized as follows.
In the next section, we
define some terms that are used throughout this paper,
and present a few basic facts.
Sec.~\ref{sec:cluster} introduces the concepts of a cluster and section
(intuitively, a bunched group of moving evacuees),
which play a key role in subsequent discussions.
In Sec.~\ref{sec:1sink},
we design an algorithm which finds a minsum 1-sink.
Sec.~\ref{sec:ksink} formulates the framework for solving the minsum $k$-sink problem,
utilizing DP.
In Sec.~\ref{sec:algorithm}
we implement a solution method to the DP formulation,
and analyze its complexity.
Finally, Sec.~\ref{sec:conclusion} concludes the paper.

\section{Preliminaries}\label{prelim.se}
\subsection{Model}\label{sec:model}
Let $P(V,E)$ denote a given path network,
where the vertex set $V$ consists of $v_1, v_2, \ldots, v_n$,
which we assume to be arranged in this order,
from left to right horizontally.
Vertex $v_i$ has weight $w_i \in \mathbb{R}^+$ (set of positive reals),
representing the number of evacuees initially located at $v_i$,
and edge $e_i=(v_i,v_{i+1})\in E$ has {\em length} or {\em distance} $d_i~(> 0)$
and {\em capacity} $c_i$,
which is the upper limit on the flow rate through $e_i$ in persons/unit time.
By $p\in P$,
we mean that point $p$ lies on $P$.
For $p,q\in P$ we write $p \prec q$ if $p$ lies to the left of $q$.
For two points $p \preceq q$,
the subpath between them is denoted by $P[p,q]$,
and $d(p,q)$ (resp. $c(p,q)$) denotes its length
(resp.~the minimum capacity of the edges on $P[p,q]$).
It takes each evacuee $\tau$ units of time to travel a unit distance on any edge.

Our model assumes that the evacuees at every vertex start evacuation at the same time
at the rate limited by the capacity of its outgoing edge.
We sometimes use the term ``cost'' to refer to the aggregate evacuation time
of a group of evacuees to a certain destination.
A {\em $k$-sink},
which means a set of $k$ sinks,
shares the following property of the {\em median} problem~\cite{kariv1979b}.
\begin{lemma}{\rm \cite{higashikawa2015a}}\label{lem:atVertex}
There is a $k$-sink such that all the $k$ sinks are at vertices.
\end{lemma}

If we plot the arrival flow rate at, or departure flow rate from, a
vertex as a function of time,
it consists of a sequence of {\em (temporal) clusters}.
The {\em duration} of a cluster is the length of time
in which the flow rate corresponding to the cluster is greater than zero.
A cluster consists of a sequence of {\em sections},
such that any adjacent pair of sections in it have different heights.
In other words, a section is a maximal part of a cluster with the same {\em height} 
(= flow rate).
A {\em simple cluster} consists of just one section.
Clearly,
in the uniform edge capacity case,
all clusters are simple.
A time interval of flow rate 0 between adjacent clusters is called a {\em gap}.
Unless otherwise specified,
we assume that evacuees arrive at vertex $v_i$ from vertex $v_{i+1}$.
The case where the evacuees move rightward can be treated symmetrically.

The  {\em head vertex} of a cluster/section is the vertex from which the evacuee corresponding to
the start time of the cluster/section originates.
The {\em offset} of a cluster with respect to vertex $v_i$ is the time until the first evacuee
belonging to the cluster arrives at $v_i$.
We say that a cluster/section {\em carries} (the evacuees from) vertex $v_i$,
if those evacuees provide flow to the cluster/section.

We define the following cost functions.
\begin{eqnarray}
\Phi_L(x) &\triangleq&
            \text {cost contribution to $x$ from~} P[v_1,v_i] \text{~if } v_i\prec x \preceq v_{i+1},\nonumber\\
\Phi_R(x) &\triangleq& \mbox{cost contribution to $x$ from~} P[v_{i+1},v_n] \text{~if } v_i\preceq x \prec v_{i+1}, \nonumber
\end{eqnarray}
and
\begin{equation}\label{eqn:Phisx} 
\Phi(x) = \Phi_L(x) + \Phi_R(x).\\
 \end{equation}
A point $x=\mu$ that minimizes $\Phi(x)$ is called a {\em minsum 1-sink}.

The total cost is the sum of the costs of all the sections.
The cost of a section of height $c$ with offset $t_0$ and duration $\delta_t$
is given by
\begin{equation}\label{eqn:blockCost}
\lambda t_0 + \frac{\lambda^2}{2c},
\end{equation}
where $\lambda = c \delta_t$ is the number of evacuees carried by the section~\cite{higashikawa2015c}.
The average evacuation time for an evacuee carried by this section is $t_0 + \lambda/2c$,
where $\lambda/2c$ represents the average delay before departure from the head vertex of the section,
and the aggregate is given by $(t_0 + \lambda/2c)\times \lambda$,
which yields \eqref{eqn:blockCost}.
The first term in \eqref{eqn:blockCost} represents the time for all evacuees carried by the section
to travel from the head vertex of the section to the destination, which is $t_0$ away,
and is called the {\em extra cost} of the section.
The second term in \eqref{eqn:blockCost} represents the time for all evacuees carried by the section
to travel from their origin vertices to the head vertex of the section,
and is called the {\em intra cost} of the section.
To be exact, the ceiling function ($\lceil~\rceil$) must be used to compute costs,
because the last group of evacuees to leave the head vertex may not occupy the full
capacity of the outgoing edge.
but we omit it for simplicity,
and {\em adopt} \eqref{eqn:blockCost} as our objective function~\cite{cheng2013}.
Or we can consider each molecule of a fluid-like material as an ``evacuee.''

A {\em minsum k-sink} partitions the path into $k$ subpaths,
and places a 1-sink on each subpath in such a way that
the cost of the max-cost 1-sink is minimized.
We shall solve the easier 1-sink problem first and the $k$-sink problem later.

\section{Cluster/section sequence}\label{sec:cluster}
In the 1-sink problem,
to compute the intra and extra costs at $v_i$,
we obviously need to know the arrival section sequence at $v_i$.
Let $\alpha_R(v_i)$ (resp. $\beta_R(v_i)$) denote the arrival section sequence
at (resp. departure section sequence from) vertex $v_i$ from {\em R}ight,
both moving left.
It is clear that $\alpha_R(v_n) = \beta_R(v_1)=\Lambda$,
where $\Lambda$ denotes the empty sequence.
Let us compute $\alpha_R(v_i)$ successively for $i=n, n-1,\ldots,1$.
If the height of an arriving section at vertex $v_i$ is higher than $c_{i-1}$,
the evacuees carried by that section cannot depart from $v_i$ at the arrival rate.
We see that the duration of the section gets stretched in this case,
by the {\em ceiling} operation~\cite{mamada2006}.
From now on, we use the verb {\em ceil} to mean performing a ceiling operation.
Moreover, 
when the first evacuee of $\alpha_R(v_i)$ arrives at $v_i$,
there may be a {\em backlog} of delayed evacuees still waiting to depart from $v_i$.
We use those delayed evacuees to fill gaps in $\alpha_R(v_i)$.
This is only for convenience for analysis purposes, 
and we actually use the latest arrivals to fill gaps.

\begin{proposition}\label{prop:juggle}
The order among the evacuees within a cluster can be rearranged arbitrarily
without affecting the cost of the cluster.
\end{proposition}
\begin{proof}
Consider evacuee $e_a$ (resp. $e_b$) who takes $t_a$ (resp. $t_b$) time to arrive at a destination.
If we interchange the positions of these evacuees,
the sum of their contributions to the total cost remains the same at $t_a+t_b$.  
\end{proof}

If the height of an arriving section at $v_i$ is less than $c_{i-1}$ (the capacity of the exit edge from $v_i$),
then we use the underutilized capacity to accommodate as many of the delayed evacuees
as possible,
together with the evacuees carried by the section.
Based on Proposition~\ref{prop:juggle},
we fix the beginning of each arriving section at its original position in the time sequence.

\begin{example}
Let $\alpha_{R}(v_i)$ consist of sections, $S_1, S_2, \ldots$.
Fig.~\ref{fig:sections2} illustrates a situation where the heights of some arriving sections
are higher than $c_{i-1}$.
\begin{figure}[ht]
\centering
\subfigure[]{\includegraphics[height=20mm]{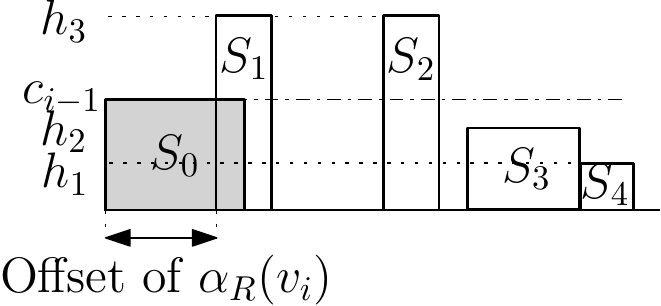}}
\hspace{80mm}
\subfigure[]{\includegraphics[height=14mm]{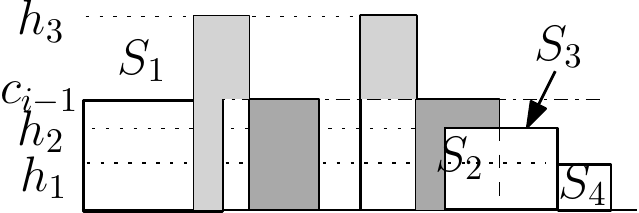}}
\hspace{3mm}
\subfigure[]{\includegraphics[height=14mm]{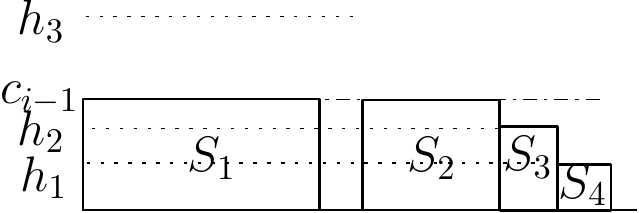}}
\caption{(a) $\alpha_{R}(v_i)$;
(b) Amount equal to the light gray parts fill the gray parts; (c) Result.
}
\label{fig:sections2}
\end{figure}
Fig.~\ref{fig:sections2}(a) shows that the evacuees at $v_i$ leave
$v_i$ at the rate of $c_{i-1}$,
forming section $S_0$,
and the first evacuee carried by the first section $S_1$ from $v_{i+1}$
arrives at $v_i$ before $S_0$ ends.
\qed
\end{example}

The above example shows that the following two situations can arise,
when $\alpha_{R}(v_i)$ is converted into $\beta_{R}(v_i)$.
\subsubsection{Observations}
\begin{enumerate}
\item[(a)]
A stretched section, due to ceiling, ends in a gap.
(Fig.~\ref{fig:sections2}(b) shows that the overlapping part of $S_0$ and $S_1$,
as well as the part of $S_1$ that exceeds capacity $c_{i-1}$ (shown in light gray)
are merged into one section (renamed $S_1$),
and get stretched in time,
partially filling the gap before $S_2$ in Fig.~\ref{fig:sections2}(a).)
\item[(b)]
A section may shrink due to the stretched section preceding it,
with its start time pushed to a later time.
(In Fig.~\ref{fig:sections2}(b) and (c), the stretched $S_2$ ``swallows'' a part of $S_3$ and $S_3$ shrinks.
The next section (such as $S_4$ in this example),
if any,
undergoes no change,
since its height is less than $c_{i-1}$.) 
\qed
\end{enumerate}

\section{1-sink problem}\label{sec:1sink}
In this section, we first show how to compute the arrival and departure sequences
$\{\alpha_R(v_i), \beta_R(v_i), \alpha_L(v_i), \beta_L(v_i)\mid i=1,\ldots,n\}$ efficiently.
Based on them,
we compute the extra and intra costs at each vertex,
which provide the costs $\{\Phi_L(v_i) + \Phi_R(v_i)  \mid 1\le i\le n\}$.
Finally, a 1-sink is found by looking for a point $x$ that minimizes
$\Phi_L(x) + \Phi_R(x)$.

\subsection{Computing section sequences}\label{sec:sectionSeq}
From Observations (a) and (b) above,
we can easily infer the following lemmas.
\begin{proposition}\label{prop:sectionCreation}
For $i=n, n-1,\ldots$,
the number of sections in $\beta_R(v_i)$ is at most
one more than that in $\alpha_R(v_i)$. 
\end{proposition}

\begin{proposition}\label{prop:clusterHeights}
For each $i~(2\le i\le n)$,
the heights of the sections in $\beta_R(v_i)$,
due to the evacuees on subpath $P[v_i,v_n]$,
are non-increasing with time.
\end{proposition}
We similarly define $\alpha_L(v_i)$ (resp. $\beta_L(v_i)$),
i.e., the arrival section sequence at 
(resp. departure section sequence from) $v_i$ from {\em L}eft, 
due the evacuees on $P[v_1, v_{i-1}]$ (resp. $P[v_1, v_i]$).
To compute $\alpha_R(v_i)$ and $\beta_R(v_i)$,
we start from vertex $v_n$.
It is easy to construct $\beta_R(v_n)$,
which consists of just a single section of height $c_{n-1}$ and duration $w_n/c_{n-1}$
that starts at time 0 (according to the local time at $v_n$),
and $\alpha_R(v_{n-1})$,
which is a gap (=offset) of length $d_{n-1}\tau$,
followed by $\beta_R(v_n)$ (according to the local time at $v_{n-1}$).
Let us consider $\alpha_R(v_i)$ and $\beta_R(v_i)$
for a general $i$, $1\le i < n$,
where $\alpha_R(v_n) = \beta_R(v_1) =\Lambda$.
There are $w_i$ evacuees who depart from $v_i$ for $v_{i-1}$ at the rate of $c_{i-1}$.
Their departure takes $w_i/c_{i-1}$ time,
and if 
$
w_i/c_{i-1} \le d_i\tau,
$
or
\begin{equation}\label{eqn:mergeCond1}
w_i/d_i\tau \le c_{i-1},
\end{equation}
then there is no interaction at $v_i$ between the cluster carrying the evacuees from $v_i$
and the first cluster in $\alpha_R(v_i)$.
In addition,
if $c_{i-1} < c_i$ then the heights (=flow rates) of some clusters in $\alpha_R(v_i)$
get reduced and their durations become stretched,
to become clusters in $\beta_R(v_i)$.
Fig.~\ref{fig:sections2} illustrates this situation.
If $w_i/d_i\tau< c_{i-1} $, on the other hand,
the first evacuee in $\alpha_R(v_i)$ arrives at $v_i$ when there is a backlog of 
evacuees from $v_i$,
still waiting for departure.

Let $\lambda(C)$ denote the total number of evacuees carried by cluster $c$.
For the $m^{th}$ cluster $C_m^i$ in $\alpha_R(v_i)$,
we record the value $\lambda(C_m^i)/t^i_m$,
where $t^i_m$ is the time difference between the start times of $C^i_m$ and
the next succeeding section $C^i_{m+1}$,
which equals $\tau$ times the distance between the first vertices of $C^i_m$ and $C^i_{m+1}$.
In what follows, we omit the superscript $i$ from these quantities,
since it will be obvious from the context.
If there is no succeeding section after $C_m$, we set $t_m=\infty$.
For each $m$,
we define the {\em critical capacity}
\begin{equation}\label{eqn:critical}
\lambda(C_m)/t_m.
\end{equation}

\begin{figure}[h]
\centering
\includegraphics[height=26mm]{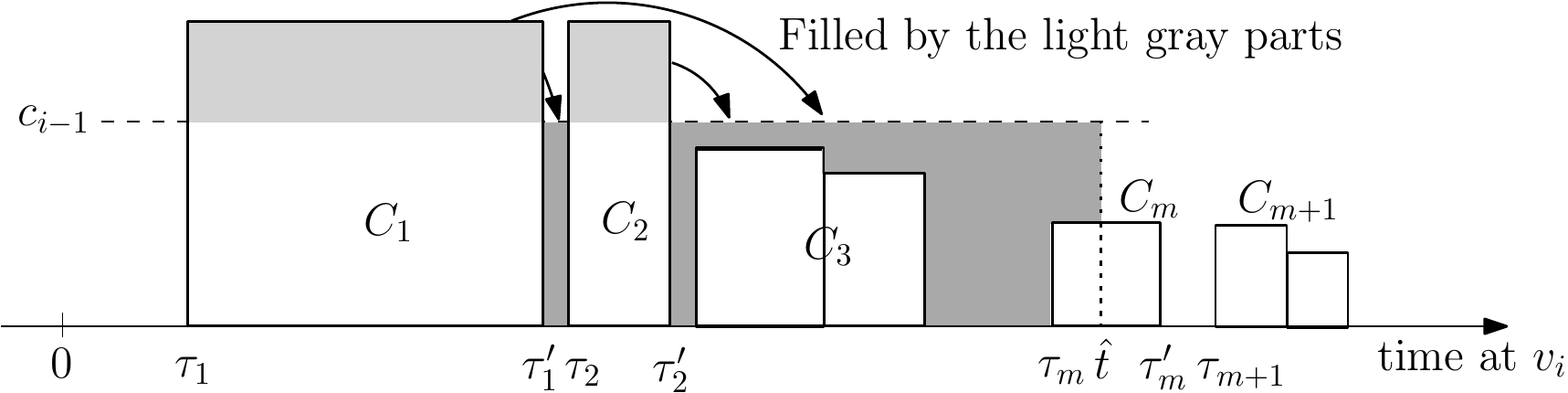}
\caption{The sum of the light gray areas equals the sum of the gray areas.
}
\label{fig:findCluster}
\end{figure}

Fig.~\ref{fig:findCluster} is similar to Fig.~\ref{fig:sections2},
except that it does not show cluster $C_0$ and contains many symbols for easy reference.
Our intention is to separate the ceiling operation from the gap-filling operation
for the backlog at $v_i$.
In the figure,
the start (resp. end) time of section $C_m$ is indicated by $\tau_m$ (resp. $\tau'_m$).
We thus have
\[
\tau_m = \sum_{j=1}^m t_j + \tau_1.
\]
The sum of the light gray areas above the capacity $c_{i-1}$ equals the sum of the gray areas,
that end in the middle of section $S_m$,
which is partially ``swallowed up'' by the gray area.
Let us now find the right end of the gray areas step by step.
Since $\lambda(C_1)/t_1>c_{i-1}$,
we merge $C_1$ and $C_2$, 
filling the gap between $C_1$ and $C_2$,
because the first evacuee carried by $C_2$ arrives at $v_i$
before the last evacuee carried by $C_1$ leaves $v_i$.

Now $C_1$ and $C_2$ have been merged into one section.
We repeat this process by remembering the ratio $(\lambda(C_1) + \lambda(C_2))/(t_1+t_2)$,
and so forth.
In general,
we find the smallest\footnote{It is not unique without the ``smallest'' condition.
So binary search cannot be used.}
$m$ such that $t=t_{m+1}$ satisfies
\begin{equation}\label{eqn:extent}
\sum_{j=1}^m\lambda(C_j)/\sum_{j=1}^m t_j \le c_{i-1}.
\end{equation}
If (\ref{eqn:extent}) still holds when $\sum_{j=1}^m t_j$ in it is replaced by
$\sum_{j=1}^m t_j + (\tau'_m-\tau_m)$,
then the gray area ends at $\hat{t}$,
where $\tau_m < \hat{t} \le \tau'_m$
(as in Fig.~\ref{fig:findCluster}),
otherwise $\tau'_m < \hat{t} \le \tau_{m+1}$.
In either case,
it is easy to determine $\hat{t}$ in (additional) constant time.

Once we determine the end of the gray area as above,
we continue with the ceiling operations for subsequent sections in $\alpha_R(v_i)$.
To facilitate it,
we introduce {\em max-heap} $\cal H$,
in which we place critical capacities of \eqref{eqn:critical},
precomputed at the time $\alpha_R(v_i)$ is constructed.
We can thus extract from $\cal H$ the largest ratio $\lambda(C_m)/t_m$.
As long as $\cal H$ contains a ratio larger than $c_{i-1}$,
we need to merge the corresponding sections.

\begin{lemma}\label{lem:secSequences}
We can compute section sequences $\{\alpha_R(v_i), \beta_R(v_i)\mid i=1,\ldots,n\}$
in $O(n\log n)$ time.
\end{lemma}
\begin{proof}
For $i=n$ we initialize max-heap $\cal H=\emptyset$,
and we update it as we compute $\{\alpha_R(v_i), \beta_R(v_i)\}$ for $i=n,n-1,\ldots$.
so that we can easily find the largest $\lambda(C_m)/t_m$ from $\cal H$ in constant time
and compare it with $c_{i-1}$.
If $\lambda(C_m)/t_m> c_{i-1}$ 
then we merge $C_m$ and $C_{m+1}$.
Since $\lambda(C_m)/t_m\ge \lambda(C_{m+1})/t_{m+1}$,
the weight-time ratio for the combined section satisfies
\[
\lambda(C_{m+1})/t_{m+1} \le \{\lambda(C_m)+\lambda(C_{m+1})\} /(t_m+t_{m+1}) \le \lambda(C_m)/t_m. 
\]
We thus put $\{\lambda(C_m)+\lambda(C_{m+1})\} /(t_m+t_{m+1})$ in $\cal H$,
which may or may not be the largest item in $\cal H$.
Each insertion into heap $\cal H$ takes $O(\log n)$ time.
Then we process the next largest item in $\cal H$ and proceed with extending/merging.
Proposition~\ref{prop:sectionCreation} implies that only $O(n)$ new sections are created in total,
as $i$ changes from $n$ to $1$,
each incurring constant processing time.
Therefore,
the total time for creating all new sections is $O(n\log n)$,
taking the insertion/extraction operations into/from $\cal H$ into account.

We now go back to deal with the overlapping area (with backlog $w$)
between sections $C_0$ and $C_1$,
shown in Fig.~\ref{fig:sections2}(a).
Let $C_1, C_2,\ldots$, be the sections after ceiling by $c_{i-1}$.
To find the end time, $\bar{t}$, of the extended $C_1$,
we first determine, using binary search, the smallest $t$ that satisfies
\begin{equation}\label{eqn:test1}
(t - \tau_1)c_{i-1}  \ge w + \sum_{m=1}^l \lambda(C_m),
\end{equation}
where $\tau'_l\le t$. 
Once such a $t$ is found then it is straightforward to compute $\bar{t}$.
This gap-filling operation takes $O(\log n)$ time per vertex $v_i$,
and the total for all vertices is $O(n\log n)$.\footnote{Sequential search is even faster,
taking $O(n)$ time in total.
But this method is useful in solving the $k$-sink problem.
}
\end{proof}

\subsection{Computing $\Phi_R(v_i)$ and $\Phi_L(v_i)$}\label{sec:extraIntra}
We discussed above how to convert $\alpha_R(v_i)$ into $\beta_R(v_i)$,
which becomes $\alpha_R(v_{i-1})$ with an offset.
We want to compute cost $R(i,n) = E_i+I_i$ for $i=n, n-1,\ldots, 1$,
where $E_i$ (resp. $I_i$) is the extra cost  (resp. intra cost) of $\alpha_R(v_i)$,
defined after \eqref{eqn:blockCost}.\footnote{$L(1,j)$ can be computed similarly.}
It is clear that  $E_n=I_n=0$.
It is also easy to see that $E_{n-1} = d_{n-1} w_n\tau$
and $I_{n-1} =  w_n^2/2c_{n-1}$.
Define an array of weights by $W[i] \triangleq \sum_{j=i}^n w_j$,
which can be precomputed in $O(n)$ time for all $i$.
By definition, we have for $i=n-1,n-2, \ldots, 1$,
\begin{equation}
E_i = E_{i+1}+ d_iW[i+1] \tau. \label{eqn:gaph1}
\end{equation}
Thus computing $E_{i-1}$ from $E_i$ takes constant time.
\begin{lemma}\label{lem:task2a}
The extra costs, $\{E_i\mid 1\le i\le n\}$, can be computed in $O(n)$ time.
\end{lemma}

Let us now consider intra cost $I_i$,
which is more difficult to compute.
All the sections of the same height belonging to different clusters are said to form a {\em height group},
or just a {\em group} for short.
They are always consecutive.
Let $\alpha_R(v_i)$ consist of a set $\cal G$ of groups of sections.
The heights of the groups are decreasing from the first group to the last group in $\cal G$.
Suppose that a group $G$ of height $h(G)$ consists of sections $S_1, S_2, \ldots, S_g$,
and,
for $q=1,\ldots, g$,
let $\lambda(S_q)$ denote the sum of the weights of the vertices carried by $S_q$.
Then the sum of the intra costs of the sections in $G$ is given by
$I(G)= \sum_{q=1}^g \lambda(S_q)^2/2h(G)$,
and the total intra cost is $I({\cal G})= \sum_{G\in {\cal G}} I(G)$.
When sections in $\alpha_R(v_i)$ are merged due to a lower capacity $c_{i-1}$,
$I(G)$ clearly changes for some groups $G$.

During the construction of $\{\alpha_R(v_i),\beta_R(v_i)\mid i=n,n-1,\ldots, 1\}$,
we need to update the intra costs of the affected section groups.
To this end,
we add the squared weight of the newly stretched section and the squared weight of
the part not swallowed up totally ($C_m$ in Fig.~\ref{fig:findCluster}).
Then subtract the sum of the squared weights of the totally or partially swallowed up sections.
We maintain $\cal G$ with the sum of squared weight for each $G\in \cal G$.
Since a swallowed up section no longer contributes to the updating time,
the total computation cost is $O(n)$.
Thus the amortized updating time for the intra cost per section is constant.

Suppose that a group $G$ of sections have the same height $c$,
and their heights get reduced due to a new, smaller capacity $c_{i-1}~(<c)$.
Let them have weights $\lambda(S_1), \ldots, \lambda(S_g) ~(g\ge 2)$,
so that the sum of their intra costs on departure from $v_i$, i.e., in $\beta_R(v_i)$,
is given by
\begin{equation}\label{eqn:intraCost}
\frac{1}{2c_{i-1}} \{\lambda(S_1)^2+ \cdots + \lambda(S_g)^2\} = \frac{\hat{\lambda}_G}{2c_{i-1}},
\end{equation}
where
\begin{equation}\label{eqn:squaredWeights}
\hat{\lambda}_G =\lambda(S_1)^2 +\cdots + \lambda(S_g)^2.
\end{equation}
This implies that we can treat all the sections in $G$ as if they were just one virtual section
with effective intra cost $\hat{\lambda}_G/2c_{i-1}$.
This way,
we can maintain just the sum $\hat{\lambda}_G$ of squared weights for $G$.
This idea works fine if sections do not merge with each other.
Assume now that $m'~(m'\le g)$ sections do merge due to a lower capacity encountered.
It implies that $m'-1$ sections disappear.
We subtract the contributions of the $m'-1$ disappeared sections from $\hat{\lambda}_G$
and replace them by the squared weight of the resulting large section.
Then the new intra cost is given by $\hat{\lambda}_G/2c_{i-1}$.
Since a section can disappear no more than once,
the total computation cost of updating $\hat{\lambda}_G$ is $O(n)$.
Clearly the intra costs of the groups that are not stretched do not change,
so they incur no updating cost.
The above argument proves the following lemma.
\begin{lemma}\label{lem:intra}
The intra costs, $\{I_i \mid 1\le i\le n\}$, can be computed in $O(n\log n)$ time.
\end{lemma}

Since $\Phi_R(v_i) = E_i+I_i$,
Lemmas~\ref{lem:secSequences}, \ref{lem:task2a},
and \ref{lem:intra} directly imply the following corollary.
\begin{corollary}\label{cor:costs}
We can compute $\{\Phi_R(v_i) \mid 1\le i\le n\}$ in $O(n\log n)$ time.
Similarly, we can compute $\{\Phi_L(v_i)  \mid 1\le i\le n\}$ in $O(n\log n)$ time.
\end{corollary}
We then compute $\min\{\Phi_L(v_i) + \Phi_R(v_i)  \mid 1\le i\le n\}$ in $O(n)$ additional time.
A vertex $v_i$ that achieves $\min\{\Phi_L(v_i) + \Phi_R(v_i)  \mid 1\le i\le n\}$ is an optimal 1-sink.
By Corollary~\ref{cor:costs},
we have
\begin{theorem}\label{thm:main3}
A minsum 1-sink in path networks (with general edge capacities)
can be found in $O(n\log n)$ time.
\end{theorem}
We formally state our algorithm as Algorithm~1 
below.

\renewcommand\footnoterule{}  
\begin{floatTogether}
\begin{algorithm}[H]\label{alg:alg-1}
{\bf Inputs:} 
$\{w_i \mid 1\le i\le n\}$ and $\{d_i, c_i\mid 1\le i\le n-1\}$\;
{\bf Outputs:} 
Optimal sink $\mu^*= v_{i^*} \in V$;
Its cost $\Phi^*$\;
\BlankLine
$\Phi_R(v_n) \leftarrow 0$\;
$\lambda(C) \leftarrow w_n$; $E \leftarrow 0$; 
\tcp{\!\!Create simple cluster $C$ carrying $v_n$. $E$ is extra cost.}
$G \leftarrow \langle C\rangle;\; h(G) \leftarrow c_{n-1};\;
\hat{\lambda}_G \leftarrow\lambda(C)^2$; \hspace{2mm}\tcp{\!\!Create first group $G$ in $\beta_R(v_n)$
with height $h(G)$ and squared weight $\hat{\lambda}_G$.}
${\cal G} \leftarrow \langle G\rangle$;  \hspace{2mm}\tcp{\!\!$\cal G$
is a sequence of groups ordered by their heights.} 
 ${\cal H} \leftarrow w_{n-1}/d_{n-1}\tau$;  \hspace{2mm}\tcp{\!\!Max-heap ${\cal H}$
 contains critical capacities.}
\For {$i \in \{n-1, \ldots, 1\}$} {
  \While {$[{\cal H} \not= \emptyset] \wedge [c_{i-1} \le {\tt top}({\cal H})]$}{
    $h \leftarrow {\tt top}({\cal H})$; ~~\tcp{\!\!${\tt top}({\cal H})$ is max item in ${\cal H}$. It is
    removed.}
 	\If {$h= \lambda(C_m)/t_m$$^{(*)}$
    \footnotetext{$^{(*)}$  See \eqref{eqn:critical}. $C'_m$ is the cluster following $C_m$.}}{
    	Merge $C_m$ and $C'_m$ to form a new section $C_p$, and let $C'_p$ be the section following it\;
    	Ceil it by $c_{i-1}$, and put $\lambda(C_p)/t_p$ in ${\cal H}$, where $t_p$ is the time between
	the start times of $C_p$ and $C'_p$\; 
    	Update each affected group $G\in {\cal G}$, and recompute $h(G)$ and $\hat{\lambda}_G$;~~\tcp{\!\!As
	in Sec.~\ref{sec:extraIntra}.}
      } 
  } 
   $E \leftarrow E+W[i+1] d_i\tau$; $\Phi_R(v_i) \leftarrow E+\sum_{G\in {\cal G}} \hat{\lambda}_G/2h(G)$\;
  }  
Compute $\Phi_L(v_i)$ for $i \in \{1, \ldots, n\}$ similarly to $\Phi_R(v_i)$\;
$i^* \leftarrow \arg \min_{1 \le i \le n}(\Phi_L(v_i) + \Phi_R(v_i))$;
$\Phi^* \leftarrow  \Phi_L(v_{i^*}) + \Phi_R(v_{i^*})$.
\caption{\sc Minsum $1$-sink}
\end{algorithm}
\end{floatTogether}

\section{$k$-sink problem}\label{sec:ksink}
In this section,
we solve the $k$-sink problem,
which entails applying a number of somewhat sophisticated methods.

\subsection{DP formulation}\label{sec:formulation}
We first present a dynamic programming (DP) formulation that follows the template
of recursive functions proposed by Hassin and
Tamir~\cite{hassin1991} for the $p$-median problem.
Our innovation consists in the manner in which we process the recursive computations
efficiently,
given that the cost functions for the sink location problem are significantly
more difficult to compute than those for the regular median problem,
due to the time-dependent nature of the evacuee flow and congestion.

Let $F^k(j)$, $1\le k\le j \le n$, denote the minsum cost when $k$ sinks are placed 
on subpath $P[v_1,v_j]$.
Similarly, define $G^k(j)$, $1\le k\le j \le n$, as the minsum cost  when $k$ sinks are placed 
on subpath $P[v_1,v_j]$,
{\em and} $v_j$ is the rightmost sink.
We can start with $j= k+1$,
since $F^k(j)=G^k(j)=0$ for $j\le k$.

Let $R(i,j)$, $i\le j$ denote the cost of evacuating from {\bf R}ight to $v_i$
all the evacuees on $P[v_{i+1}, v_j]$.
Similarly, let $L(i,j)$, $i\le j$ denote the cost of evacuating from {\bf L}eft to $v_j$
all the evacuees on $P[v_i,v_{j-1}]$.
More generally,
let us define $F^p(j)$ and $G^p(j)$ for $1\le p\le k$.
By definition, we have 
\begin{eqnarray}
F^p(j) &=& \min_{p\le i \le j} \{G^p(i) + R(i,j)\}, \label{eqn:recursive1} \\
G^p(j) &=& \min_{p\le i \le j} \{F^{p-1}(i) + L(i+1,j)\}.\label{eqn:recursive2}
\end{eqnarray}
To initialize the above recursive computations,
we use $\{F^1(j) \mid 1\le j\le n\}$,
which can be computed from $\{\Phi_L(v_i), \Phi_R(v_i) \mid 1\le i\le n\}$ given
by Corollary~\ref{cor:costs}.\footnote{We could also start with $\{G^1(j)\mid 1\le j\le n\}$,
since $G^1(j)=\Phi_L(v_j)$,
which can be obtained from Corollary~\ref{cor:costs}.
}
We can then compute $G^2(j)$, using \eqref{eqn:recursive2},
and $F^2(j)$, using \eqref{eqn:recursive1}, and so forth.
We can thus compute $\{G^p(j),F^p(j)\mid p=2,\ldots, k\}$,
provided $R(i,j)$ and $L(i,j+1)$ are readily made available.
Moreover,
to obtain a DP algorithm with time complexity sub-quadratic in $n$,
we also need to quickly find the index $i$ that minimizes the recurrence relations
\eqref{eqn:recursive1} and \eqref{eqn:recursive2}.

Let us address the latter issue first.
We refer to the evaluation of \eqref{eqn:recursive1} and \eqref{eqn:recursive2} as {\em phase} $p$
computation.
To visualize our approach to computing \eqref{eqn:recursive1},
for a given phase $p~(2\le p \le k)$,
let us plot points $(G^p(i),R(i,j))$ in a 2-dimensional coordinate system
for all $i~(1\le i\le j)$,
for a fixed vertex $v_j$.
See Fig.~\ref{fig:R-G}.
\begin{figure}[ht]
\centering
\includegraphics[width=7.3cm]{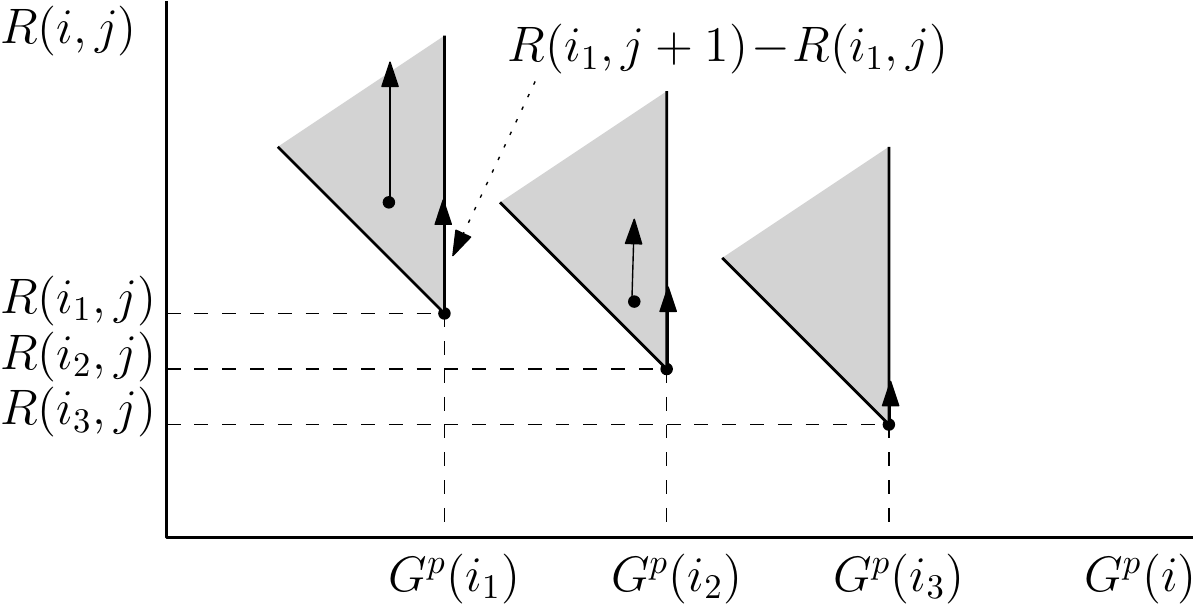}
\caption{$R(i,j)$ vs. $G^p(i)$ for fixed $j$.
}
\label{fig:R-G}
\end{figure}
If we superimpose the line represented by $G^p(i) + R(i,j)=c$ for a given value $c$
in the same coordinate system,
it is a $-45^{\circ}$ line.
If we increase $c$ from 0,
this line eventually touches one of the plotted points.
It is clear that the first point it touches gives the optimal value for $i$ that minimizes $G^p(i) + R(i,j)$.
In Fig.~\ref{fig:R-G},
this optimal is given by the point $(G^p(i_1), R(i_1,j))$,
hence $F^p(j) = G^p(i_1)+R(i_1,j))$.
For convenience, let us refer to point $(G^p(i), R(i,j))$ as {\em point $(i,j)$}.

We now explain that this representation provides us very useful information.
To see it,
for each point $(i,j)$,
define the {\em $V$-area} that lies above the $-45^{\circ}$ line and
to the left of the vertical line through it as shown as shaded areas in Fig.~\ref{fig:R-G}. 
We say that a point $(i',j)$ situated in the $V$-area of another point
$(i,j)$ is {\em dominated} by $(i,j)$, 
since the cost of point $(i',j)$ is higher than the cost of $(i,j)$.
In what follows, we identify a vertex $v_i$ with its index $i$,
and may say vertex $i$ to refer to $v_i$.
We sometimes say that $v_{i'}$ is dominated by $v_i$,
when $j$ is clear from the context.
Thus the points at the bottoms of the V-areas are the only non-dominated points.
For subpath $P[v_1,v_j]$ let $I^p[j] = \{i_1, \ldots, i_{g(j)}\}$,
where $i_1 < i_2 < \ldots < i_{g(j)}\le j$ and $\{(i,j)\mid i\in I^p[j]\}$
are the set of all points at the bottoms of the V-areas.
As observed above,
we have
\begin{proposition}\label{prop:Ii}
For $p=2,3,\ldots, k$,
$F^p(j) = G^p(i_1) + R(i_1, j)$ holds.
\end{proposition}
Function $G^p(j)$ can be computed in a similar manner,
based on \eqref{eqn:recursive2}.
Since $i_{g(j)}\le j$, 
vertex $v_{j+1}$ is farther from $v_{i_s}$ than it is from $v_{i_t}$,
if $s<t$.
We thus have
\begin{equation}
R(i_s,j+1)-R(i_s,j) \ge R(i_t,j+1)-R(i_t,j) \mbox{~~ for~} s<t.
\end{equation}
The upward arrows in Fig.~\ref{fig:R-G} indicate the increase $R(i,j+1)-R(i,j)$
for different $i$'s.
Moreover,
if $(i,j)$ is dominated by $(i_s,j)$,
then point $(i,j')$ will also be dominated by $(i_s,j')$ for any $j'>j$. 
This implies that once it is determined that $(i,j)\notin I^p[j]$,
then $(i,j')$ will not belong to $I^p(j')$ for any $j'>j$.
We will discuss how to update $I^p[j]$ to $I^p[j+1]$,
as $j$ increases,
 in the next subsection. 

\subsection{Computing switching points}\label{sec:switchingPoint}
To compute $F^p(j)$ by Proposition~\ref{prop:Ii},
we maintain index set $I^p[j]$ of non-dominated vertices.
Let us denote by $x(i_{s-1},i_s)$,  $1 < s \le g(j)$,
the {\em switching point},
namely the leftmost vertex $j' ~(g(j) < j' \le n)$, if any, for which $i_s$ dominates $i_{s-1}$. 
If such an index does not exist, it means that
$i_s$ never dominates $i_{s-1}$ and therefore we need not remember $i_s$.
Formally, we define, for $s \ge 2$,
\begin{equation}\label{eqn:switchPt}
x(i_{s-1}, i_{s}) = 
\min \{j' >i_s \mid G^p(i_s)+R(i_s,j') \le G^p(i_{s-1})+R(i_{s-1},j')\}.
\end{equation}
Define a sequence of switching points
\begin{equation}\label{eqn:Xp}
X^p[j] \triangleq \begin{cases}
	&\langle x(i_1, i_2), \ldots,x(i_{g(j)-1},i_{g(j)})\rangle \text{ if }|I^p[j]|\ge 2\\
	&\Lambda \hspace{43mm}\text{ if }|I^p[j]|=1.
\end{cases}
\end{equation}
Maintaining $X^p[j]$ together with $I^p[j]$ allows us to update $I^p[j-1]$ easily to $I^p[j]$.
\begin{lemma}\label{lem:switchingPts}
The switching points in $X^p[j]~(\not=\Lambda)$ satisfy
\begin{equation}\label{eqn:switchPt2}
x(i_1, i_2) < x(i_2,i_3) < \cdots < x(i_{g(j)-1},i_{g(j)}).
\end{equation}
\end{lemma}
\begin{proof}
Assume for example that $x(i_b,i_c) < x(i_a,i_b)$ holds,
where $i_a<i_b<i_c$.
Then $i_b$ will never be an optimal vertex, 
because for large enough $j~(\ge x(i_a,i_b))$ which makes $v_{i_b}$ dominate $v_{i_a}$,
vertex $v_{i_c}$ already dominates $v_{i_b}$,
since $x(i_b,i_c) < x(i_a,i_b)$.
This implies that $i_b$ should not be in $I^p[j]$.
\end{proof}

Procedure {\sc Add-Vertex}($v_j$) updates $I^p[j-1]$ to $I^p[j]$ for any phase 
$p=2,\ldots, k$.
The {\bf while} clause is justified by Lemma~\ref{lem:switchingPts}.
\begin{procedure} \label{proc:proc-1}
    \SetKwInOut{Input}{Input Data}
    \SetKwInOut{Output}{Output}
\Input{$I^p[j-1]= \{i_1, \ldots, i_{g(j-1)}\}, X^p[j-1]$; $\{G^p(i) \mid 1\le i\le n\}$\;}
\Output{$I^p[j]$, $X^p[j]$; \tcp*{\!\!$j \ge 2$ assumed}}
\BlankLine
$I^p[j]\leftarrow I^p[j-1]$, $X^p[j]\leftarrow X^p[j-1]$\;
\If {$i_2$ exists and $x(i_1, i_2) \le j$}{
Remove $i_1$ from $I^p[j]$ and $x(i_1, i_2)$ from $X^p[j]$;~\tcp{\!\!$x(i_2, i_3) > j$ if $i_3$ exists.
}
}
Set $i$ to the last index in $I^p[j]$\;
\While{$i^-($index immedaiately before $i$ in $I^p[j])$ exists and $x(i,j)\le x(i^-,i)$}{
Remove $i$ from $I^p[j]$, and remove $x(i^-,i)$ from $X^p[j]$\;
Set $i=i^-$\;
}
Append $x(i,j)$ to $X^p[j]$\;
Append $j$ to $I^p[j]$.
\caption{Add-Vertex($v_j$)}
\end{procedure}
Let $t_X(n)$ denote the time needed to compute $x(i,i')$ for an arbitrary pair
$(i,i')$, $i < i'$.
\begin{lemma}\label{lem:Ij+1}
Running Procedure {\sc Add-Vertex}$(v_j)$ for all $j=2,\ldots, n$ takes $O(nt_X(n))$ time.
\end{lemma}
\begin{proof}
We need to compute $x(\cdot,\cdot)$ in Line 6 of {\sc Add-Vertex}($v_j$).
It is computed once whenever a vertex is entered or removed from an index set,
thus $O(n)$ times in total.
\end{proof}

Let $t_R(n)$ denote the time required to compute $R(i,j)$ 
for an arbitrary pair $(i,j)$.
Then we have
\begin{lemma}\label{lem:xii'2}
$t_X(n)=O(t_R(n) \log n)$.
\end{lemma}
\begin{proof}
To compute $x(i,i')$,
we perform binary search for the pivot $j'$ (in \eqref{eqn:switchPt}) in the range $\{i',\ldots, n\}$.
Each probe requires evaluating $R(\cdot,\cdot)$.
\end{proof}
By the above lemma,
we are naturally interested in making $t_R(n)$ as small as possible,
which is the topic of Sec.~\ref{sec:Rij}.

\section{Algorithms}\label{sec:algorithm}
\subsection{Statement}\label{sec:statement}
Algorithm~2, 
given below,
computes the minsum $k$-sink.
It requires a procedure to compute $R(i,j)$,
which will be discussed in Sec.~\ref{sec:Rij}.
It maintains index list, $I^p[j]$, and a list of switching points, $X^p[j]$,
in each phase $p$.

\begin{lemma}\label{lem:main:ksink1}
The minsum $k$-sink in path networks with {\em general} edge capacities
can be computed in $O(kn \cdot t_X(n))$ plus preprocessing time.
\end{lemma}
\begin{proof}
In each phase $p$ ($1\le p \le k$),
the most time consuming operation is the updating of $I^p[j]$ and $X^p[j]$.
Algorithm~2 performs $O(kn)$ iterations in Lines 2 and 4.
The {\bf for} loop (Line 7) invokes {\sc Add-Vertex}($v_j$) $O(n)$ times.
In these invocations, an element is added to $I^p[~]$ at most once,
so $I^p[~]$ cannot receive more than a total of $n$ elements.
Thus $x(\cdot, \cdot)$ is computed $O(n)$ times in each repetition of the
outer {\bf for} loop (Line 2),
costing $O(n \cdot t_X(n))$ time by  Lemmas~\ref{lem:Ij+1} and \ref{lem:xii'2}.
Similarly Line 9 costs $O(n \cdot t_X(n))$ per repetition of the outer {\bf for} loop.
The lemma follows from Lemma~\ref{lem:Ij+1},
since all other steps need less time.
\end{proof}

\setcounter{algocf}{1}
\renewcommand\footnoterule{} 
\begin{floatTogether}
\begin{algorithm}[H]\label{alg:alg-2}
    \SetKwInOut{Input}{Input}
    \SetKwInOut{Output}{Output}
\Input{$\{w_i\in V\}$; $\{c_i, d_i\mid e_i \in E\}$\;}
\Output{
Set $S^* \subseteq V$ of $k$ sinks;
Cost $\Phi^*$ of solution $S^*$\;}
\BlankLine
Compute $\{G^1(j)\mid j =1, \ldots, n\}$\tcp*{\!\!$G^1(j)=\Phi_L(v_j)$.}
\For {$p \in \{1, \ldots, k\}$} { 
  $I^p \leftarrow \langle1\rangle; X^p \leftarrow \Lambda$\tcp*{\!\!Initialize.}
  \For {$j \in \{1, \ldots, p\}$} {
  $F^p(j) =0$\; 
  }
  \For {$j \in \{p+1, \ldots, n\}$} {  
        Invoke {\sc Add-Vertex}($v_j$)\;
    Set $F^p(j) \leftarrow G^p(i_1) + R(i_1, j)$\;
  } 
  \eIf {$p < k$}{
     Compute $\{G^{p+1}(j)\mid1 \le j \le n\}$ in a similar way using $\{F^p(j)\mid1 \le j \le n\}$}{
     \Return $\Phi^* = F^k(n)$ \!\! \tcp*{\!\!Sink set $S^*$ can be obtained from $\Phi^*$
     by isolating a maximum subpath with a sink with cost $\le \Phi^*$ at a time from left.}
  }
}
\caption{\sc Minsum $k$-sink algorithm}
\end{algorithm}
\end{floatTogether}

\subsection{Computing $R(i,j)$ and $L(i,j)$}\label{sec:Rij}
Line 6 of Algorithm~2 uses $R(i,j)$,
based on Eq. \eqref{eqn:recursive1},
and Line 12 uses $L(i,j)$,
based on Eq.\eqref{eqn:recursive2}.
We only discuss below how to compute $R(i,j)$,\
which takes $t_R(n)$ time by definition,
since $L(i,j)$ can be computed in similar time, $t_L(n)$.

\subsubsection{Preprocessing}
~

As preprocessing before computing $R(i,j)$,
we construct a binary tree $\cal T$,
called the {\em cluster sequence tree},
whose leaves are the vertices of $P$,
arranged from $v_1$ to $v_n$.
For node\footnote{We use the term ``node'' for $\cal T$
to distinguish them from the vertices of $P$.  The leaf nodes of
$\cal T$ are the vertices of $P$.
}
$u$ of $\cal T$ ,
let $v_L(u)$ (resp. $v_R(u)$) denote the leftmost (resp. rightmost) vertex
that belong to subtree ${\cal T}(u)$.
We say that $u$ {\em spans} subpath $P[v_L(u),v_R(u)]$.
Let $\alpha_R([v_i,v_j])$ (resp. $\beta_R([v_i,v_j])$) denote the arrival section sequence
from {\bf R}ight at
(resp. departure section sequence from) vertex $v_i$, moving left,
carrying the evacuees initially located at the vertices on $P[v_{i+1},v_j]$ (resp. $P[v_i,v_j]$).
Define
\begin{align}
\alpha_R(u) &\triangleq\alpha_R([v_L(u),v_R(u)]), \nonumber\\
\beta_R(u) &\triangleq\beta_R([v_L(u),v_R(u)]).\nonumber
\end{align}
At each node $u$ of $\cal T$,
we store $\beta_R(u)$.
We call two nodes $u_a$ and $u_b$ of $\cal T$ {\em adjacent} if $v_R(u_a)$ and $v_L(u_b)$
are adjacent vertices on $P$.
\begin{proposition}\label{prop:lastCluster}
Let $u_a$ and $u_b$ be two adjacent nodes of $\cal T$.
The evacuees still left at vertex $v_L(u_a)$ when the first evacuee in $\beta_R(u_b)$ arrives there
belong to the last cluster in $\beta_R(u_a)$.
\end{proposition}\begin{proof}
Let the last cluster mentioned in the proposition start with the evacuees from vertex $v$.
Then clearly the first evacuee from $v$ in $\beta_R(u_a)$ must have left $v$
before the first evacuee from $v_L(u_b)$ arrives there.
\end{proof}

\begin{lemma}\label{lem:tau1}
We can construct ${\cal T}$ in $O(n\log n)$ time.
\end{lemma}
\begin{proof}
We construct ${\cal T}$ bottom up.
For a leaf node $u$ (a single vertex $v_i$ of $P$) of $\cal T$,
we have $\alpha_R(u)=\emptyset$,
and $\beta_R(u)$ consists of a section of height $c(v_{i-1})$ and duration $w_i/c(v_{i-1})$.

Let $u_a$ and $u_b$ be the two child nodes of $u$ in ${\cal T}$,
and assume that $\beta_R(u_a)$ and $\beta_R(u_b)$ are already available.
Note that the first cluster of $\beta_R(u_b)$ arrives at $v_L(u_a)$ with
a delay of $d(v_L(u_a),v_L(u_b))\tau$.
Let $v_i = v_L(u_a)$, $v_l = v_R(u_a)$, and $u=u_b$ in Fig.~\ref{fig:merge5}.
Let $C$ be the last cluster of $\beta_R(u_a)$,
which starts with vertex $v$, 
and let $c=c(v_{i-1}, v_L(u_b))$.
\begin{figure}[h]
\centering
\includegraphics[height=29mm]{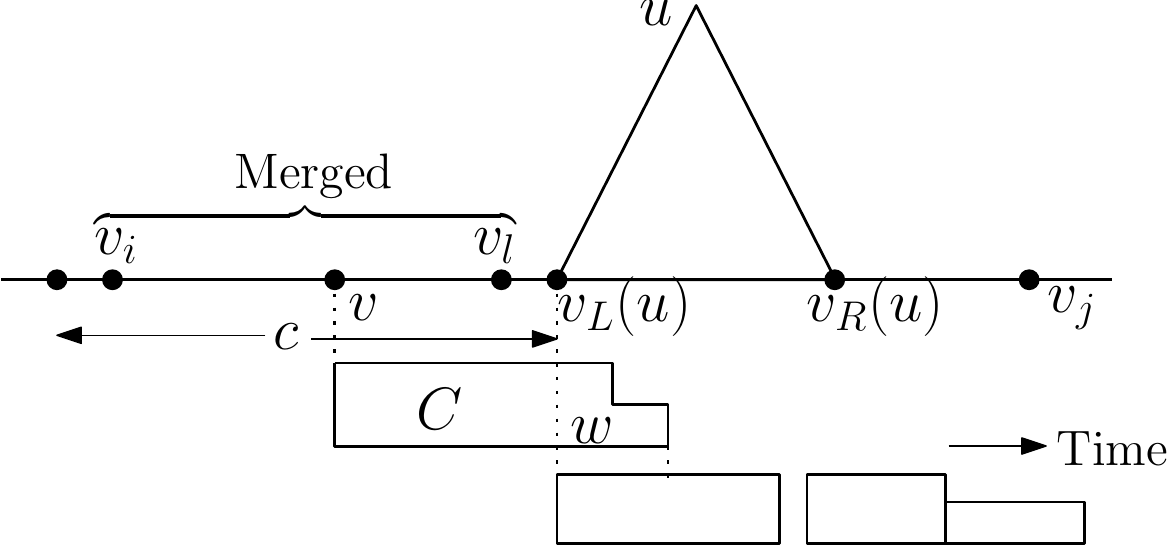}
\hspace{2mm}
\caption{Merging a new subtree ${\cal T}(u)$.
}
\label{fig:merge5}
\end{figure}
We first need to ceil $\beta_R(u_b)$ by $c$,
to merge $C$ and initial sections of $\beta_R(u_b)$,
assuming that they overlap with a backlog of $w$,
as shown in Fig.~\ref{fig:merge5}.\footnote{We can find the backlog in constant time,
using Proposition~\ref{prop:lastCluster}.
}
To this end we can use a max-heap $\cal H$,
as in Algorithm~1. 
Initially,
we place $O(n)$ critical capacities of the form \eqref{eqn:critical} in $\cal H$,
one for each cluster in $\beta_R(u_b)$.
We remove the max value from $\cal H$,
and if it is larger than $c$,
merge the corresponding two clusters.
Every time two clusters merge,
a new critical capacity is inserted into $\cal H$.
See Proposition~\ref{prop:merge}.
Since a cluster disappears after every merger,
the total number of insertions into $\cal H$ is $O(n)$,
costing a total of $O(n\log n)$ time.

After the above ceiling operation,
we need to deal with the backlog of $w$.
As discussed in Sec.~\ref{sec:sectionSeq},
using the condition (\ref{eqn:test1}),
we can use binary search to find the extent of the stretched cluster in the ceiled $\beta_R(u_b)$.
It is clear that the total time for all critical capacities is $O(n\log n)$.
\qed
\end{proof}

\begin{proposition}\label{prop:merge}
Let $C_m$ be the cluster corresponding to the max value from ${\cal H}$.
(Thus $\lambda(C_{m+1})/t_{m+1}\le \lambda(C_m)/t_m$ holds.)
Ceiling them by $\lambda(C_m)/t_m$,
$C_m$ and $C_{m+1}$ are merged into one cluster with
critical capacity $(\lambda(C_m)+\lambda(C_{m+1})/(t_m+ t_{m+1})$,
which satisfies
\[
\lambda(C_{m+1})/t_{m+1}\le(\lambda(C_m)+\lambda(C_{m+1})/(t_m+ t_{m+1})
\le \lambda(C_m)/t_m\lambda(C_m).
\qed
\] 
\end{proposition}
By the above proposition,
if we ceil $\beta_R(u_b)$ by $\lambda(C_m)/t_m$,
then the critical capacity for the merged cluster is less than $\lambda(C_m)/t_m$,
which is inserted in ${\cal H}$.

\subsubsection{Computing $R(i,j)$}
~

To compute $R(i,j)$,
we need to know the arriving cluster/section sequence at $v_i$
of the evacuees from the vertices on $P[v_{i+1}, v_j]$.
To discuss this problem,
let ${\cal P}[v_{i+1}, v_j]$ denote the set of maximal subpaths of $P[v_{i+1}, v_j]$
spanned by $t=O(\log n)$ nodes,
$u_1, u_2, \ldots, u_t$, of $\cal T$,
in this order from left to right.
We combine the departure section sequences stored at  these nodes
into a single sequence $\alpha_R([v_i,v_j])$.\footnote{Recall that $\alpha_R([v_i,v_j])$ does not
include the evacuees from $v_i$.
Thus it is obtained by shifting $\beta_R([v_{i+1},v_j])$.
}
Starting with $\sigma=\beta_R(u_1)$,
we update $\sigma$ by merging it with shifted $\beta_R(u_2), \beta_R(u_3)$, $\ldots$,
until all of them are merged into one section sequence.
The left part of Fig.~\ref{fig:merge5} shows the subpaths merged so far (current $\sigma$),
and the right triangle shows the next subtree ${\cal T}(u)$ to be merged to $\sigma$.
The shift amount for $\beta_R(u_s)$ is $d(v_L(u_1), v_L(u_s))\tau$.
When we merge $\sigma$ with $\beta_R(u_s)$,
the capacity $c(v_{i-1}, v_L(u_s))$ ($c$ in Fig.~\ref{fig:merge5}) must be used
to ceil the shifted $\beta_R(u_s)$.

As in the proof of Lemma~\ref{lem:tau1},
there are the following two main tasks in computing the departure
sequences and their costs.
\begin{enumerate}
\item[{[T1]}]
Implement ceiling operations efficiently.
\item[{[T2]}]
Implement the spreading of the backlog efficiently.
\end{enumerate}
We can deal with Task T1,
using a max-heap $\cal H$, as we did in the proof of Lemma~\ref{lem:tau1}.
However, we need to use the ceiling operations,
not only in preprocessing, 
but also in later processing.
Therefore, we need a method by which we can ``reuse'' results of precomputation
repeatedly.
Our main interest is in solving \eqref{eqn:test1} for the end time $\bar{t}$.
But it cannot be done efficiently (i.e., by binary search)
unless Task1 is carried out first.
To this end we construct a forest as follows.

\subsubsection{Cluster forest ${\cal F}_u$}\label{sec:mergeTree}
Arrange the nodes representing the clusters in $\beta_R(u)$ horizontally
from left to right in their order in $\beta_R(u)$.
They are the leaf nodes of the trees in ${\cal F}_u$.
At each leaf node, we store data about the sections (start time, end time, and its height)
in the cluster it represents.
When two clusters merge,
we introduce their parent node,
and record the critical capacity that caused the merger,
as well as the data on the constituent sections.
Here we can use a max-heap $\cal H$ again.
We put a newly computed weight/time ratio in $\cal H$,
provided it is not less than $c_0 =\min\{c_i\mid 1\le i\le n-1\}$.
The construction of ${\cal F}_u$ completes when $\cal H$ becomes empty.
Proposition~\ref{prop:merge} implies the following.
\begin{proposition}\label{prop:capacityOrder}
On the path from a leaf to the root of any tree in the forest ${\cal F}_u$,
the critical capacities recorded at the nodes visited are non-increasing.
\end{proposition}
Based on the above proposition,
we can use Procedure {\sc Find-cluster}($c,v_h$) below 
to solve the following problem:
Given a capacity $c$ and a vertex $v_h$,
find the cluster that contains $v_h$ when $\beta_R(u)$ is ceiled by a given capacity $c$.
\begin{floatTogether}
\begin{procedure}[H]\label{proc:proc-3}
    \SetKwInOut{Input}{Input data}
    \SetKwInOut{Output}{Output}
\Input{${\cal F}_u$\;}
\Output{Cluster that contains $v_h$ when $\beta_R(u)$ is ceiled by $c$
}
\BlankLine
Find the tree $T$ in ${\cal F}_u$ that contains $v_h$ in a leaf node,
and let $c'$ be the capacity stored there\;
\While {$c\le c'$}{
Move to the parent node $q'$ and let $c'$ be the capacity stored at $q'$\;
 }
Let $q$ be the child node of $q'$ visited just before $q'$.
Output the cluster stored at $q$.
\caption{Find-cluster($c,v_h$)}
\end{procedure}
\end{floatTogether}
This procedure is useful for Task T2,
since we can do binary search on vertices to find $\bar{t}$ by solving \eqref{eqn:test1}.

\subsubsection{Skip lists}
Procedure {\sc Find-Cluster}($c,v_h$), as stated, is very simple,
but its major drawback is that it requires linear time to run.
To make it sublinear,
we introduce {\em skip lists}~\cite{pugh1990}.
Consider any tree $T$ in ${\cal F}_u$.
We rearrange its nodes in such a way that the path from any node
to the farthest leaf node in the left subtree is not shorter than that in the right subtree.
We first construct the skip lists for the path $\pi_1$ from the leftmost leaf node to the root
in the resulting tree.
We then construct the skip lists for the path $\pi_2$ from the second leftmost leaf node to the root.
From the node of $T$,
where $\pi_1$ and $\pi_2$ meet,
to the root,
we let $\pi_2$ share the skip lists which were already constructed for $\pi_1$.
To construct this data structure more systematically,
we start at the root of $T$,
and place pointers at nodes which are at distances at every 2 nodes, 4 nodes, 8 nodes, etc.,
according to the guideline in \cite{pugh1990}.
This implies that $O(\log n)$ {\em levels} of skip lists are constructed.
In each skip list at a level,
pointers point to nodes that are closer to the root on the path to the root.\footnote{The pointers
in a skip list at a level skip the same number of nodes.
}
Finally, we place a pointer from a leaf node to the first node of the pointer chain
in each skip list.
It is easy to see that the total number of pointers is $O(n)$,
and placing them takes $O(n\log n)$ time.

\begin{lemma}\label{lem:useSkipList}
If the skip lists are superimposed on ${\cal F}_u$ as above,
we can find the cluster that contains $v_h$ when $\beta_R(u)$ is ceiled by $c$
in $O(\log n)$ time.
\end{lemma}

\subsubsection{Intra costs}

Let us now consider the intra cost component of $R(i,j)$.
The intra costs must be updated whenever two clusters get merged by a critical
capacity out of $\cal H$.
It can be done during prepocessing as in Lemma~\ref{lem:intra},
and we store the updated intra cost 
for the entire subpath $P[v_L(u), v_R(u)]$
at the corresponding node in ${\cal F}_u$ that stores the merged cluster.

\begin{lemma}\label{lem:Icost4forest}
We can compute the intra costs at all the nodes of ${\cal F}_u$ in time linear in the number
of leaf vertices in ${\cal T}(u)$.
\end{lemma}

\begin{lemma}\label{lem:nonuniformRij2}
If $\cal T$ (with $\beta_R(u)$ and ${\cal F}_u$ at every node $u$) and the skip lists for
${\cal F}_u$ are given,
then we have $t_R(n)=O(\log^3 n)$.
\end{lemma}
\begin{proof}
We need to identify the head vertex of the cluster that carries vertex $v_h$
that is probed by binary search for capacity $c=c(v_{i-1}, v_L(u_s))$.
This can be done in $O(\log n)$ time by Lemma~\ref{lem:useSkipList}.
and it takes $O(\log^2 n)$ time for all the $O(\log n)$ probes.
We need to repeat this for each of the $O(\log n)$ subtrees of $\cal T$ that spans
a subpath in ${\cal P}[v_{i+1}, v_j]$.
\end{proof}

\subsection{Main results}
The correctness of our approach follows from the fact that we are solving the DP problem
formulated by \eqref{eqn:recursive1} and \eqref{eqn:recursive2},
which are clearly correct,
and the correctness of Algorithm~2, 
which is an implementation of \eqref{eqn:recursive1} and \eqref{eqn:recursive2},
whose correctness we have argued in our discussions so far.
From Lemmas~\ref{lem:xii'2}, \ref{lem:main:ksink1}, \ref{lem:nonuniformRij2}, and \ref{lem:uniform},
we obtain our final results.
\begin{theorem}\label{thm:main.ksink1}
The minsum $k$-sink problem in path networks with non-uniform
edge capacities can be solved in $O(kn\log^4 n)$ time.
\end{theorem}

\subsubsection{Uniform capacity case}\label{sec:uniform}
Assume that all the edges have the same capacity $c$,
so that every cluster (=section) has height $c$.
In this case, there is no need for ceiling operations,
hence no need for forest ${\cal F}_u$.

\begin{lemma}\label{lem:uniform}
\begin{enumerate}
\item[(a)]
We can construct $\beta_R(u)$ at each node $u$ of $\cal T$,
and compute its extra and intra costs in $O(n\log n)$ time.
\item[(b)]
If $\cal T$ is given,
then we have $t_R(n) =O(\log^2 n)$.
 \end{enumerate}
\end{lemma}
\begin{proof}
(a) 
We can construct $\beta_R(u)$ in $O(n\log n)$ time by Lemma~\ref{lem:tau1}. 
We can also compute the intra costs bottom up,
as in our algorithm for the 1-sink problem discussed in Sec.~\ref{sec:extraIntra}.

(b) 
The most time consuming operation is carrying Task2 in each of $O(\log n)$ mergers,
which takes $O(\log n)$ time per merger.
\end{proof}

\begin{theorem}\label{thm:main.ksink2}
The minsum $k$-sink problem in path networks with uniform edge capacities
can be solved in
$O(kn\log^3 n)$ time.
\end{theorem}

\section{Conclusion}\label{sec:conclusion}
We proposed efficient algorithms based on DP that find a minsum $k$-sink
in path networks.
They run in $O(kn\log^4n)$ and $O(kn\log^3n)$
if the edge capacities are non-uniform and uniform,
respectively.
An open problem is to efficiently solve the minsum $k$-sink problem in networks
that are more general than path networks.
In a straightforward way,
we can find a minsum 1-sink in tree networks in $O(n^2\log^2n)$ time,
applying the $O(n\log^2n)$ time algorithm in \cite{mamada2006} that computes
the {\em arriving table} at each vertex.
We believe we can extend our results to cycle networks with a small increase 
in time complexity.

\section*{Acknowledgement}\label{sec:acknowledgement}
This work is supported in part by NSERC of Canada Discovery Grants,
awarded to Robert Benkoczi and Binay Bhattacharya,
in part by JST {\sc Crest} (JPMJCR1402),
and in part by JSPS {\sc Kakenhi} Grant-in-Aid for Young Scientists (B) (17K12641).


\end{document}